\documentclass[letterpaper, 10 pt, conference]{ieeeconf}  % Comment this line out
\IEEEoverridecommandlockouts                              % This command is only
\usepackage{amsmath}
\usepackage{amssymb,amsfonts}
\usepackage{graphicx}
\usepackage{tikz}

\usepackage{cite}
\newtheorem{lemma}{\textbf{Lemma}}{}%[section]
\newtheorem{definition}{\textbf{Definition}}{}%[section]

\title{\LARGE \bf
Stabilizing Error Correction Codes \\ for Controlling LTI Systems over  Erasure Channels
}

\author{Jan {\O}stergaard% <-this % stops a space
\thanks{J. {\O}stergaard is with the Section on Artificial Intelligence and Sound, Department of Electronic Systems,  Aalborg University, 9220 Aalborg, Denmark.
        {\tt\small jo@es.aau.dk}}%
}

\begin{document}

\maketitle
\thispagestyle{empty}
\pagestyle{empty}

%%%%%%%%%%%%%%%%%%%%%%%%%%%%%%%%%%%%%%%%%%%%%%%%%%%%%%%%%%%%%%%%%%%%%%%%%%%%%%%%
\begin{abstract}

We propose $(k,k')$ stabilizing codes, which is a type of delayless error correction codes that are useful for control over networks with erasures.
For each input symbol, $k$ output symbols are generated by the stabilizing code. Receiving any $k'$ of these outputs guarantees stability. Thus, the system to be stabilized is taken into account in the design of the erasure codes. Our focus is on LTI systems, and we construct codes based on independent encodings and multiple descriptions. The theoretical efficiency and performance of the codes are assessed, and their practical performances are demonstrated in a simulation study. There is a significant gain over other delayless codes such as repetition codes. 

\end{abstract}

%%%%%%%%%%%%%%%%%%%%%%%%%%%%%%%%%%%%%%%%%%%%%%%%%%%%%%%%%%%%%%%%%%%%%%%%%%%%%%%%
\section{INTRODUCTION}

There has been a vast amount of literature on networked control systems over erasure channels, cf. \cite{tatikonda:,sinopoli:2004,jin:2006,imer:2006,sahai:2006,liu:2007,quevedo:2007,ostrovsky:2009,gupta:2009,trivellato:2010,elia:2011,quevedo:2011,garone:2012,yuksel:2013,quevedo:2014,nagahara:2014,2015IJC,ostergaard:2016,peters:2016,maas:2016,khina:2019,barforooshan:2020}. In \cite{sinopoli:2004},  it was shown that for a given unstable linear time invariant (LTI) system, there exists a critical limit on the packet dropout rate beyond which the system cannot be stabilized in the usual mean-square sense.  To go beyond this critical limit, several techniques have been proposed ranging from error correction codes~\cite{sahai:2006,ostrovsky:2009} and multiple descriptions~\cite{ostergaard:2016}  to packetized predictive control~\cite{quevedo:2007} to name a few.
%In \cite{quevedo:2011,quevedo:2014,nagahara:2014,ostergaard:2016,peters:2016,barforooshan:2020}, packetized predictive control schemes were considered in order to counteract packet erasures. MORE DESCRIPTIONS OF THE OTHER WORKS.

Assume the output of the plant is to be encoded and transmitted over a digital erasure channel, where packets are either completely lost or received without errors. 
To recover from erasures,  error correction codes can be utilized \cite{singleton:1964,ostrovsky:2009}. Error-correction codes are  often designed with a certain loss rate of the channel in mind, and do not necessarily  rely on the plant (exceptions include the work in \cite{ostrovsky:2009} which tracks the plant state). For example, $(n,k)$ erasure channel codes, take $k$ \emph{source} packets and outputs $n$ \emph{channel} packets. If any $k$ of the channel packets are received, the original $k$ source packets can be completely recovered. If more than $k$ packets are received, the additional received data packets are not useful since they do not contain any further information about the plant state than what is already known. Finally, if less than $k$ packets are received, the source packets can generally not be recovered at all and all the transmitted information is in this case wasted. 
%\begin{figure}[th]
%\begin{center}
%\includegraphics{system_diagram}
%\caption{A noisy LTI system P that is controlled over a digital channel \cite{silva:2016}.}
%\label{fig:system}
%\end{center}
%\end{figure}

\begin{figure}
\centering
\scalebox{1}{
\begin{tikzpicture}

\draw   
   (2,3) rectangle (3,2)
	node[pos=0.5]{$P$};

\draw [->] (1.5,2.8) -- (2,2.8)
	node[pos=0.1,above]{$d$};
\draw [->] (1.5,2.5) -- (2,2.5)
	node[pos=0.1,left]{$x_0$};
\draw [->] (0,2.2) -- (2,2.2)
	node[pos=0.1,above]{$u$};
\draw [-] (3,2.2) -- (5,2.2)
	node[pos=0.9,above]{$y$};
\draw [->] (3,2.8) -- (3.5,2.8)
	node[pos=0.9,above]{$e$};

\draw [->] (5,2.2) -- (5,1) -- (4.7,1);
\draw [-] (0,2.2) -- (0,1) -- (0.3,1);

\draw   
   (0.3,0.7) rectangle (1.3,1.3)
	node[pos=0.5]{\footnotesize Dec}; 

\draw   
   (3.7,0.7) rectangle (4.7,1.3)
	node[pos=0.5]{\footnotesize Enc};

\draw [<-] (1.3,1) -- (1.9,1);
\draw [<-] (3.1,1) -- (3.7,1);

\draw   
   (1.9,0.7) rectangle (3.1,1.3)
	node[pos=0.5]{\footnotesize Channel};

\end{tikzpicture}
}

\caption{Noisy LTI system $P$ that is controlled over a digital channel \cite{silva:2016}.}
\label{fig:system}
\end{figure}
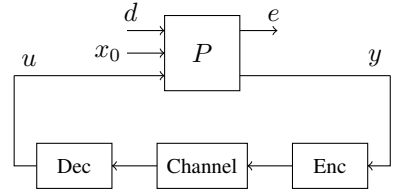

An alternative to error correction codes are multiple descriptions \cite{gamal:1982}, which combines source and channel coding. With multiple descriptions, the source is encoded into a number of descriptions, which are individually transmitted over the channel. There is no priority on the descriptions, and any subset of the descriptions can be jointly decoded to achieve a desired performance. 
Multiple descriptions were, for example, used for state-estimation in \cite{jin:2006} and combined with packetized predictive control in~\cite{ostergaard:2016}. One of the problems with multiple descriptions is that it is generally very hard to design good multiple-description codes. Another problem is that the descriptions generally contains redundant information except in the limit of vanishing data rates or when used in the extreme asymmetric situation, where the descriptions are prioritized and a successive refinement scheme is obtained. If one is able to construct a successive refinement source coder, then it was shown in \cite{yeung:1996} and \cite{puri:1999}, that the layers in the successive refinement code can be combined with traditional error correction codes in order to obtain a (sub optimal) multiple-description code. It was recently shown that a combination of successive refinement and multiple descriptions with feedback becomes   rate-distortion optimal under certain asymptotical conditions \cite{ostergaard:2021}.

We will in this paper focus on discrete-time LTI plants, stationary Gaussian disturbances, Gaussian initial state, scalar-valued control inputs and sensor outputs. Thus, the plant state can have an arbitrary dimensionality but the control signal as well as the output of the plant are both scalar valued. For such a system, the minimal information rate required to guarantee stability and a desired performance (measured in terms of the variance of the plant output) was completely characterized in \cite{silva:2016} for the case of commmunications over error-free digital channels. An illustration of the system is shown in Fig.~\ref{fig:system}.

We show that simple stabilizing erasure codes can be obtained from properly designed independent encodings~\cite{ostergaard:2021} or multiple descriptions~\cite{gamal:1982}. Specifically, for a given LTI plant we design a $(k,k')$ stabilizing code such that when combining any $k'$ descriptions of the code, the resulting $\mathrm{SNR}$ is above a critical limit, which guarantees that the decoded control signal contains sufficient information to stabilize the plant. We show that  simple codes based on independent encodings are asymptotically efficient for nearly stable plants. In general, for unstable plants, it is advantageous to use a design based on multiple descriptions. In a simulation study, we demonstrate that for the same sum-rate and delay, it is possible to achieve a significant gain in performance over that which is possible with repetition coding. 

%\begin{figure}[t]
%\begin{center}
%\includegraphics{LTI_system}
%\caption{Linear system that models the system of Fig.~\ref{fig:system}.}
%\label{fig:linear_system}
%\end{center}
%\end{figure}

\begin{figure}
\centering
\scalebox{1}{
\begin{tikzpicture}

\draw   
   (2,3) rectangle (3,2)
	node[pos=0.5]{$P$};

\draw [->] (1.5,2.8) -- (2,2.8)
	node[pos=0.1,above]{$d$};
\draw [->] (1.5,2.5) -- (2,2.5)
	node[pos=0.1,left]{$x_0$};
\draw [->] (0,2.2) -- (2,2.2)
	node[pos=0.1,above]{$u$};
\draw [-] (3,2.2) -- (5,2.2)
	node[pos=0.9,above]{$y$};
\draw [->] (3,2.8) -- (3.5,2.8)
	node[pos=0.9,above]{$e$};

\draw [->] (5,2.2) -- (5,1) -- (4.7,1);
\draw [-] (0,2.2) -- (0,1) -- (0.3,1);

\draw   
   (0.3,0.7) rectangle (1.3,1.3)
	node[pos=0.5]{$F$}; 

\draw   
   (3.7,0.7) rectangle (4.7,1.3)
	node[pos=0.5]{$L$};

\draw [<-] (1.3,1) -- (2.3,1)
	node[pos=0.4,below]{$w$};
\draw [<-] (2.7,1) -- (3.7,1)
	node[pos=0.6,below]{$v$};

\draw [->] (2.51,0.4) -- (2.51,0.8)
	node[pos=0.1,left]{$q$};

\draw   
	(2.1,1.4) rectangle (2.9,1.8)
	node[pos=0.5]{$z^{-1}$}; 

\draw (2.5, 1) circle (6pt)
	node[]{$+$};

\draw [->] (1.6,1) -- (1.6,1.6) -- (2.1,1.6);

\draw [->] (2.9,1.6) -- (4.2,1.6) -- (4.2,1.3);

\end{tikzpicture}
}

\caption{Linear system that models the system of Fig.~\ref{fig:system}.}
\label{fig:linear_system}
\end{figure}
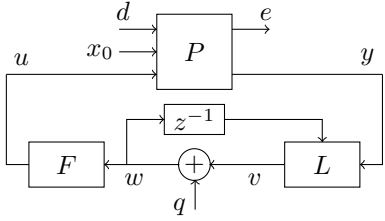

%%%%%%%%%%%%%%%%%%%%%%%%%%%%%%%%%%%%%%%%%%%%%%%%%%%%%%%%%%%%%%%%%%%%%%%%%%%%%%%%
\section{BACKGROUND}

Let us begin by  considering the  networked control system presented in \cite{silva:2016}, and which is shown in Fig.~\ref{fig:system}. Here $P$ is an LTI plant that is open-loop unstable, $u$ is the scalar control input, and $y$ is the scalar  sensor output of the plant. The external disturbance is denoted by $d$ and $e$ is an error signal that is related to the output performance. The plant output $y$ is to be encoded by the causal encoder $\mathrm{Enc}$, transmitted over the ideal noise-less digital channel, and then decoded by the causal decoder $\mathrm{Dec}$. The encoder-decoder pair $(\mathrm{Enc},
\mathrm{Dec})$ also contains the controller. Thus, the output of the decoder is the control signal to  the plant. For a fixed data rate of the coder, the performance will be measured by the variance $\sigma_e^2$ of the output $e$. We have the following linear input-output relationship through the plant $P$:
\begin{equation}
\begin{bmatrix}
e \\
y
\end{bmatrix} 
= 
\begin{bmatrix}
P_{11} & P_{12} \\
P_{21} & P_{22} 
\end{bmatrix}
\begin{bmatrix}
d \\ u
\end{bmatrix}.
\end{equation}

It was shown in \cite{silva:2016}, that if the initial state $x_0$ and the external disturbances are arbitrarily colored but jointly Gaussian, then the optimal encoder-decoder pair constitute a linear system + noise. This implies that the system in Fig.~\ref{fig:system} can be modelled by the \emph{linear} system shown in Fig.~\ref{fig:linear_system}. In this system, $F$ and $L$ are LTI systems, and $q$ is an additive white Gaussian noise, which simulates the coding noise due to source coding. In this equivalent form, we have the following relationship \cite{silva:2016}:
\begin{equation}
u = Fw, \quad w= v+q, \quad v = L  
\begin{bmatrix}
z^{-1}w \\
y
\end{bmatrix}, L = [L_w, L_y],
\end{equation}
where $z^{-1}$ indicates a one-step delay operator.
The signal-to-noise ratio $\gamma$ of the system is defined as:
\begin{equation}
\gamma \triangleq \frac{\sigma_v^2}{\sigma_q^2}.
\end{equation}

It was shown in \cite{silva:2016}, that for any proper LTI filters $F$ and $L$ that makes the system in Fig.~\ref{fig:linear_system} internally stable and well-posed, we have the following explicit expressions:
\begin{align} \label{eq:gamma}
\gamma &= \|S - 1\|^2 + \frac{1}{\sigma_q^2}\| L_y P_{21} S\|^2 \\ \label{eq:perf}
\sigma_e^2 &= \|P_{11} + P_{12}K(1-P_{22}K)^{-1}P_{21}\|^2 + \|P_{12}FS\|^2 \sigma_q^2 \\ \label{eq:S}
S &= (1-L_w z^{-1} - P_{22} F L_y)^{-1}\\
K &= F L_y(1 - L_w z^{-1})^{-1}.
\end{align}
To find the optimal filters $(F,L)$ that minimizes $\sigma_e^2$ subject to a constraint on $\gamma$, one needs to solve a convex optimization problem \cite{silva:2016}. A lower bound on the minimal coding rate $R$ achievable when using optimal filters $(F,L)$ is given by:
\begin{equation}
R \geq \frac{1}{2}\log_2( 1 + \gamma).
\end{equation}
It is clear from \eqref{eq:gamma} that asymptotically as $\sigma_q^2\to \infty$, $\gamma \to \|S-1\|^2$, which shows that the minimum $\mathrm{SNR}$ required for stability is $\|S-1\|^2$, and the minimum rate required for stability is given by $\frac{1}{2}\log_2(1 + \|S-1\|^2)$. %To get a desired performance, one  need to use a coding noise $q$ with a finite variance $\sigma_q^2$. 

\section{CAUSAL CODERS}
The encoder $\mathcal{E}_i  : \mathbb{R}^i \times \mathcal{S}^i \to \mathcal{Y}^k$ at time $i$ is a (possibly) time-varying causal one-to-many map, which at each time instance
produces $k$ outputs, that is:
\begin{equation}
(y_c^{(1)}(i), \dotsc , y_c^{(k)}(i)) = \mathcal{E}_i (y^i, s^i), 
\end{equation}
where $y_c^{(j)}(i) \in \mathcal{Y}_c$ denotes the $j$th output of the encoder at time $i$, and $y^i = y_1,\dotsc, y_i$ indicates that the encoder is only using the sequence of current and past plant outputs. The sequence $s^i$ denotes side information. Thus, the encoder can be randomized via the side information, which for example allows one to obtain a stochastic encoder. The outputs of the encoder are discrete. However, by use of subtractive dithering techniques, the resulting reconstructed values at the decoder are continuous. With this, the quantizer can be modelled as an additive white noise source \cite{zamir:2014}. 

Let $\mathcal{I}(i) \subseteq \{1,\dotsc, k\}$ denote the set of indices of the received descriptions at time $i$. 
At each time instance, $k$ descriptions are produced and transmitted over the digital erasure channels. 
%As an example, let $k=3$ and assume that the second description, $y_c^{(2)}(i)$, is lost at time $i$. Then $\mathcal{I}(i) = \{1,3\}$. Let $\mathcal{I}^i = \mathcal{I}(1), \mathcal{I}(2), \dotsc, \mathcal{I}(i)$. We are now in a position to define 
The set of causal decoders at time $i$ is $\mathcal{D}^{\mathcal{I}^i}_i :  \mathcal{Y}_c^{\mathcal{I}^i}  \times \mathcal{S}^{\mathcal{I}^i} \to \mathbb{R},\ \forall \mathcal{I}^i \subseteq \{1,\dotsc, k\}^i$. For a particular choice of decoder, say $\mathcal{D}^{\mathcal{I}^i}_i$, the reconstructed signal $u(i) \in \mathbb{R}$ at time $i$ is given by:
\begin{align}
u(i) = \mathcal{D}^{\mathcal{I}^i}_i (  y_c^{\mathcal{I}^i} , s^{\mathcal{I}^i}). 
\end{align}

Section IV considers lower bounds on the coding rates based on Gaussian coding schemes. The operational data rates obtained when using a practical coding scheme is generally greater than these lower bounds. These operational issues regarding the stabilizing codes are treated in the longer version of the paper \cite{ostergaard:2021b}. In particular, since we are here focusing on the situation with a scalar output, we need to use scalar quantizers. It is well known that scalar quantizers suffers from at rate-loss compared to vector quantizers except at very low bit rates. In addition, we need to entropy encode the output of the quantizer to further reduce the bitrate. Since the entropy coder is operating on one sample at a time, it will generally not be possible to reach the entropy of the output.

%%%%%%%%%%%%%%%%%%%%%%%%%%%%%%%%%%%%%%%%%%%%%%%%%%%%%%%%%%%%%%%%%%%%%%%%%%%%%%%%
\section{STABILIZING ERROR CORRECTION CODES}
%\section{Stabilizing Error Correction Codes}

We will first introduce some definitions, which we will be needing in the sequel. 

%\vspace{2mm}

\begin{definition}
We will denote by $(F,L,P,\gamma)$ a linear system on the form shown in Fig.~\ref{fig:linear_system}, which has coding rate $R=0.5\log_2(1+\gamma)$ and performance $D=\sigma_e^2$, where $\sigma_e^2$ is given by \eqref{eq:perf}.
\hfill$\triangle$
\end{definition}

%\vspace{2mm}

\begin{definition}
 A $(k,k')$ stabilizing code for the system $(F,L,P, \gamma)$ produces $k$ descriptions such that using any $k'$ of them is sufficient to stabilize the system.
\hfill$\triangle$
\end{definition}

%\vspace{2mm}

% the resulting $\mathrm{SNR}$, say  $\gamma'$, is greater than $\|S-1\|^2_2$, where $S$ is given in \eqref{eq:S}. 

%When using a $(k,k')$ stabilizing code, it follows that controlling the plant using any combinations of $k'$ or more descriptions is guaranteed to stabilize the plant.

To quantify the efficiency of a $(k,k')$ stabilizing code when used on a particular system $(F,L,P,\gamma)$, we will compare the sum-rate $R_S$ of the $k$ descriptions to the rate $R$ required for a single-description code to achieve the same performance as that obtained when using all $k$ descriptions (without erasures). In the linear Gaussian case, the efficiency can be assessed by simple means as shown in the definition below.

\begin{definition}
The efficiency $\eta$ of a $(k,k')$ stabilizing code for the system $(F,L,P,\gamma)$ is defined as:
\begin{equation}\label{eq:eta}
\eta \triangleq  
\frac{ \log_2( 1 + \tilde{\gamma})}{k\log_2( 1  + \hat{\gamma})},\quad 0\leq \eta \leq 1,
\end{equation}
where $\hat{\gamma}$ is the $\mathrm{SNR}$ when using any single description out of the $k$ descriptions, and $\tilde{\gamma}$ is the $\mathrm{SNR}$ when combining all $k$ descriptions. \hfill$\triangle$
\end{definition}

%\vspace{2mm}

%\textbf{Remark:} 
When measuring efficiency in \eqref{eq:eta}, we need to make sure that we compare the coding rates of systems having similar performance (in terms of $\sigma_e^2$). The best performance of a $(k,k')$ stabilizing code is obtained when using all $k$ descriptions, which results in an $\mathrm{SNR}$ of $\tilde{\gamma}$. The rate of each description is $0.5\log_2(1+\hat{\gamma})$ and the sum-rate is $0.5k\log_2(1+\hat{\gamma})$. On the other hand, when not using a stabilizing code we need a coding rate of $0.5\log_2(1+\tilde{\gamma})$ to achieve an $\mathrm{SNR}$ of $\tilde{\gamma}$. 
%It can readily be shown that the efficiency is bounded as: $0\leq \eta \leq 1$. 

For a classical $(n,k)$ error correction code that produces $n$ outputs for each $k$ input sample (or block of samples), the efficiency is $k/n$, and the delay is $k-1$  samples (blocks). A repetition code that duplicates the same source block $k$ times has efficiency $1/k$ and zero delay. The stabilizing codes that we propose are also delayless and are able to improve upon the efficiency of repetition codes  due to the property that descriptions can synergistically improve upon each other. 

%\vspace{2mm}

\subsection{Stabilizing codes based on independent encodings}

\begin{definition} \label{def:encodings}
Let $w_i = v + q_i, i = 1,\dotsc, k$. If $w_j$ and $w_i, i\neq j,$ are conditionally independent given $v$, then we refer to $w_1,\dotsc, w_k$ as \emph{independent encodings}~\cite{ostergaard:2021}.
\end{definition}

%\vspace{2mm}

\begin{lemma}
$\mathbf{(k,k')}$ \textbf{Stabilizing Code Based on Independent Encodings.}
\emph{Consider the system $(F,L,P,\gamma)$, which is illustrated in Fig.~\ref{fig:linear_system}. Let $v$ be Gaussian and let $w_i = v + q_i, i = 1, \dotsc, k$, be $k$ independent encodings of $v$,  
where $q_i, i=1,\dotsc, k,$ are mutually independent,  zero-mean Gaussian distributed, and having a common variance $\sigma^2$. If for some $1\leq k' \leq k$, the common variance satisfies
\begin{equation}\label{eq:code1}
\sigma^2\leq  \frac{\gamma k'\|L_yP_{21}S\|^{2}}{ \|S-1\|^{2} (\gamma - \|S-1\|^2)},
\end{equation}
where $S$ is given in \eqref{eq:S}, then $w_i, i=1,\dotsc, k$, form a $(k,k')$ stabilizing code for the system $(F,L,P,\gamma)$. }
\end{lemma}

%\vspace{2mm}

\begin{proof}
The variance of $\frac{1}{k'}(w_{i_1} + \cdots + w_{i_{k'}})$ is $(k')^{-1}\sigma^2$ for any subset of $k'$ encodings. The resulting $\mathrm{SNR}$  $\gamma' = k' \sigma_v^2 \sigma^{-2}$, when combining $k'$ descriptions, needs to satisfy:
\begin{align} \label{eq:1}
\gamma' = \frac{k'\sigma_v^2}{\sigma^2} > \|S-1\|^2,
\end{align}
since $\|S-1\|^2$ is the minimal $\mathrm{SNR}$ required to guarantee stability. We now use that $\gamma \sigma_q^2= \sigma_v^2$, and from \eqref{eq:gamma} we get:
\begin{align}\label{eq:sv}
 \sigma_v^2 =  \gamma (\gamma - \|S-1\|^2)^{-1}\|L_yP_{21}S\|^{2}.
\end{align}
Inserting into \eqref{eq:1} and re-arranging terms leads to:
\begin{align}\label{eq:sn2}
\sigma^2 < \gamma k' \|S-1\|^{-2} (\gamma - \|S-1\|^2)^{-1}\|L_yP_{21}S\|^{2},
\end{align}
which leads to \eqref{eq:code1}. 
\end{proof}

%\vspace{2mm}

The following lemma provides a lower bound on the sum-rate required for a $(k,k')$ stabilizing code based on independent encodings. We note that if one is not interested in the performance when receiving less than $k'$ descriptions, then the sum-rate can generally be further reduced by use of distributed source coding techniques such as Slepian-Wolf coding \cite{slepian:1973}. However, at low coding rates, the bound becomes asymptotically optimal as is shown by Lemma~\ref{lem:eta_indp}. 

\vspace{2mm}

\begin{lemma}
\emph{The minimum sum-rate $R_S$ of a $(k,k')$ stabilizing code based on independent encodings for the system $(F,L,P,\gamma)$ is:}
\begin{equation}
R_S \geq \frac{k}{2}\log_2\bigg(1 + \frac{\|S+1\|^2}{k'}\bigg).
\end{equation}
\end{lemma}

%\vspace{2mm}

\begin{proof}
Let $\sigma^2$ be the variance of the coding noise for a single description of the $(k,k')$ stabilizing code. Then, the resulting variance when linearly combining $k'$ descriptions is $\sigma^2/k'$. Thus, 
$\mathrm{SNR} = k\sigma_v^2/ \sigma^2 \geq  \|S-1\|^2$, where the inequality follows since $\|S-1\|^2$ is the minimum $\mathrm{SNR}$ that guarantees stability. Isolating $\sigma^2$ leads to:
\begin{equation}\label{eq:s2lb}
\sigma^2 \leq \|S-1\|^{-2} k' \sigma_v^2. 
\end{equation}
We can now express the sum-rate in terms of $\sigma^2$, that is:
\begin{align*}
R_S = \frac{k}{2}\log_2(1 + \frac{\sigma_v^2}{\sigma^2}) 
\geq \frac{k}{2}\log_2(1 + (k')^{-1}\,\|S-1\|^{2}).
\end{align*}
\end{proof}

%\vspace{2mm}

\begin{lemma}\label{lem:eta_indp}
\emph{Consider the system $(F,L,P,\gamma)$. The efficiency of a minimum sum-rate $(k,k')$ stabilizing code based on independent encodings is given by:
\begin{equation}
\eta =  \frac{\log_2(1 + k(k')^{-1}\,\|S-1\|^{2}) }
{k\log_2(1 + (k')^{-1}\,\|S-1\|^{2})},
\end{equation}
and the code is asymptotically efficient in the sense of:}
\begin{equation}
\lim_{\|S-1\|^2 \to 0} \eta = 1.
\end{equation}
\end{lemma}

\begin{proof}
The first part follows immediately from \eqref{eq:s2lb}, since the $\mathrm{SNR}$ for a single description is $\sigma_v^2/\sigma^2$ and for $k$ descriptions it is $ k\sigma_v^2/\sigma^2$. The second part 
  follows since the logarithm of the number $1+kc$ is approximately linear in $k$ when $c\ll 1$, i.e., $\log(1+kc) \approx k\log(1+c)$ for small $c$. 
\end{proof}

For a $(F,L,P,\gamma)$ system, if $\|S-1\| = 0$ it means that the system is stable. Thus, the second part of Lemma~\ref{lem:eta_indp} considers the situation where the plant is either stable or nearly stable, i.e., the unstable poles are near the unit circle. In this case, the coding rates are arbitrary small, and the $k$ descriptions of the $(k,k')$ stabilizing code becomes mutually independent. Thus, there is no redundancy by using $k$ descriptions each of rate $R/k$ over a single description of rate $R$  \cite{ostergaard:2021}.

%\vspace{2mm}

\begin{lemma}
\emph{Consider the system $(F,L,P,\gamma)$. The performance (in terms of $\sigma_e^2$) for this system when using $\ell \geq k'$ descriptions of a minimum sum-rate $(k,k')$ stabilizing code based on independent encodings is:}
\begin{align}\notag
\sigma_e^2 &= \|P_{11} + P_{12}K(1-P_{22}K)^{-1}P_{21}\|^2 \\
&\quad +
 \frac{k'\, \gamma  \|P_{12}FS\|^2 \|L_yP_{21}S\|^{2}}{\ell\, \|S-1\|^{2} (\gamma - \|S-1\|^2) }
 			,\quad \ell=k', \dotsc, k.
\end{align}
\end{lemma}

%\vspace{2mm}

\begin{proof}
Follows from \eqref{eq:perf} by inserting \eqref{eq:sn2} and the fact that the noise variance satisfies $\frac{\sigma^2}{\ell}$ for $\ell=1,\dotsc, k$. 
\end{proof}

\subsection{Example 1}
Consider a plant $P$ that provides the following input-output relationship between $(u,d)$ and $y$:
\begin{equation}
y = \frac{0.165}{(z-2)(z-0.5789)}( u + d), \quad e=y,
\end{equation}
where the external disturbance $d$ has a standard normal distribution. Notice that the plant has an unstable pole at $z=2$, which implies that the minimum $\mathrm{SNR}$ required for stability is $\|S-1\|^2 = 3$, and equivalently the minimum coding rate is $0.5\log_2(1 + 3) = 1$ bit. For this plant, we can choose a particular $\gamma$ and find the optimal filters $L,F$ and associated $S,K$ by using the method described in \cite{silva:2016}. From these we can find the performance $\sigma_e^2$ using \eqref{eq:perf} and coding rate $R = 0.5\log_2(1+ \gamma)$. 
Changing $\gamma$ leads to another set of $L,F,K,S$ variables and different performances and coding rates.

%Assume we are interested in a performance of $\sigma_e^2 = 4.59$, then the resulting lower bound on the coding rate can be shown to be $R=1.24$ bits. 

Let us now design a $(4,2)$ stabilizing code, so that receiving any $2$ descriptions implies that the minimum $\mathrm{SNR}$ requirement is fulfilled. We choose $\sigma^2$ so that the resulting $\mathrm{SNR}$ is $\frac{\ell}{4}\gamma$, when linearly combining $\ell=1,\dotsc, 4$ descriptions. For $\gamma = 7.2$, we obtain:
$\mathrm{SNR} = 1.8, 3.6, 5.4$, and $7.2$ for $\ell=1,\dotsc, 4$, respectively. For $\ell =2$ it is clear that the resulting $\mathrm{SNR}$ is greater than the minimum of $3$, and we therefore have a $(4,2)$ stabilizing code. 

The coding rate per description is $0.5\log_2(1 + 1.8) = 0.74$ bits, and the sum-rate is $R_S = 2.96$ bits. The coding rate required for a single-description system to achieve $\mathrm{SNR} = 7.14$ is $R=0.5\log_2( 1 + 7.14) = 1.51$ bits. Thus, the efficiency is $\eta = 1.51/2.96 = 0.51$. For comparison, a repetition code with 4 descriptions would have an efficiency of $\eta = 0.25$. 
 
In Fig.~\ref{fig:SNR}, we have illustrated the resulting $\mathrm{SNR}$ when combining the $\ell=1,\dotsc, 4$ descriptions as a function of $\gamma$. It can be seen that a $(4,1)$ stabilizing code can be obtained for $\gamma>10$ dB, and a (4,2) stabilizing code for $\gamma>7$ dB. Fig.~\ref{fig:efficiency} shows the efficiency as a function of $\gamma$. The efficiency is monotonically decreasing in $\gamma$. At very low $\gamma$ which corresponds to small coding rates, the efficiency is highest.

\begin{figure}[t]
\begin{center}
\includegraphics[width=7cm]{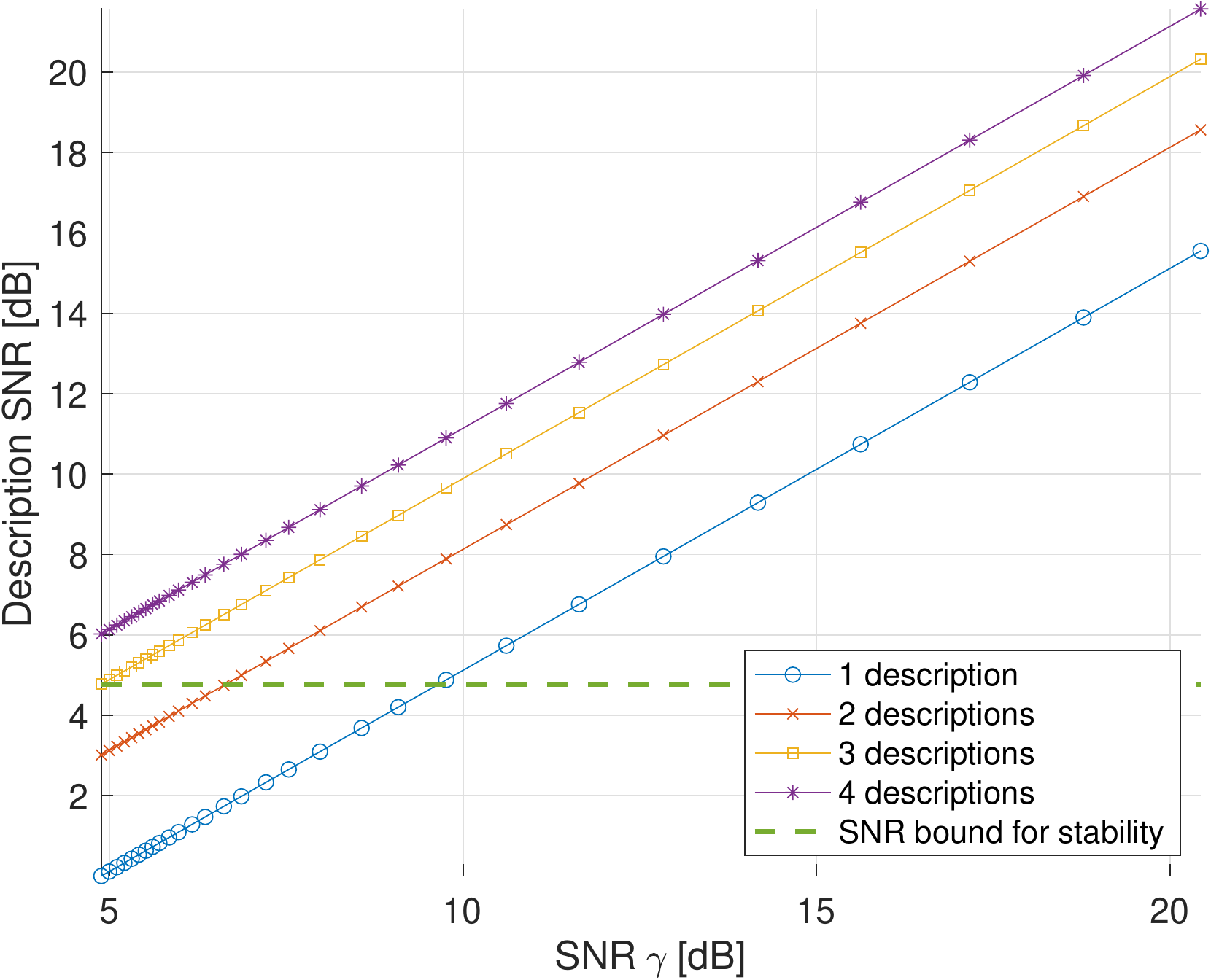}
\caption{The resulting $\mathrm{SNR}$ obtained with the $(4,2)$ stabilizing code of Example 1, when combining $\ell = 1,\dotsc, 4$ descriptions.}
\label{fig:SNR}
\end{center}
\end{figure}

\begin{figure}[t]
\begin{center}
\includegraphics[width=6cm]{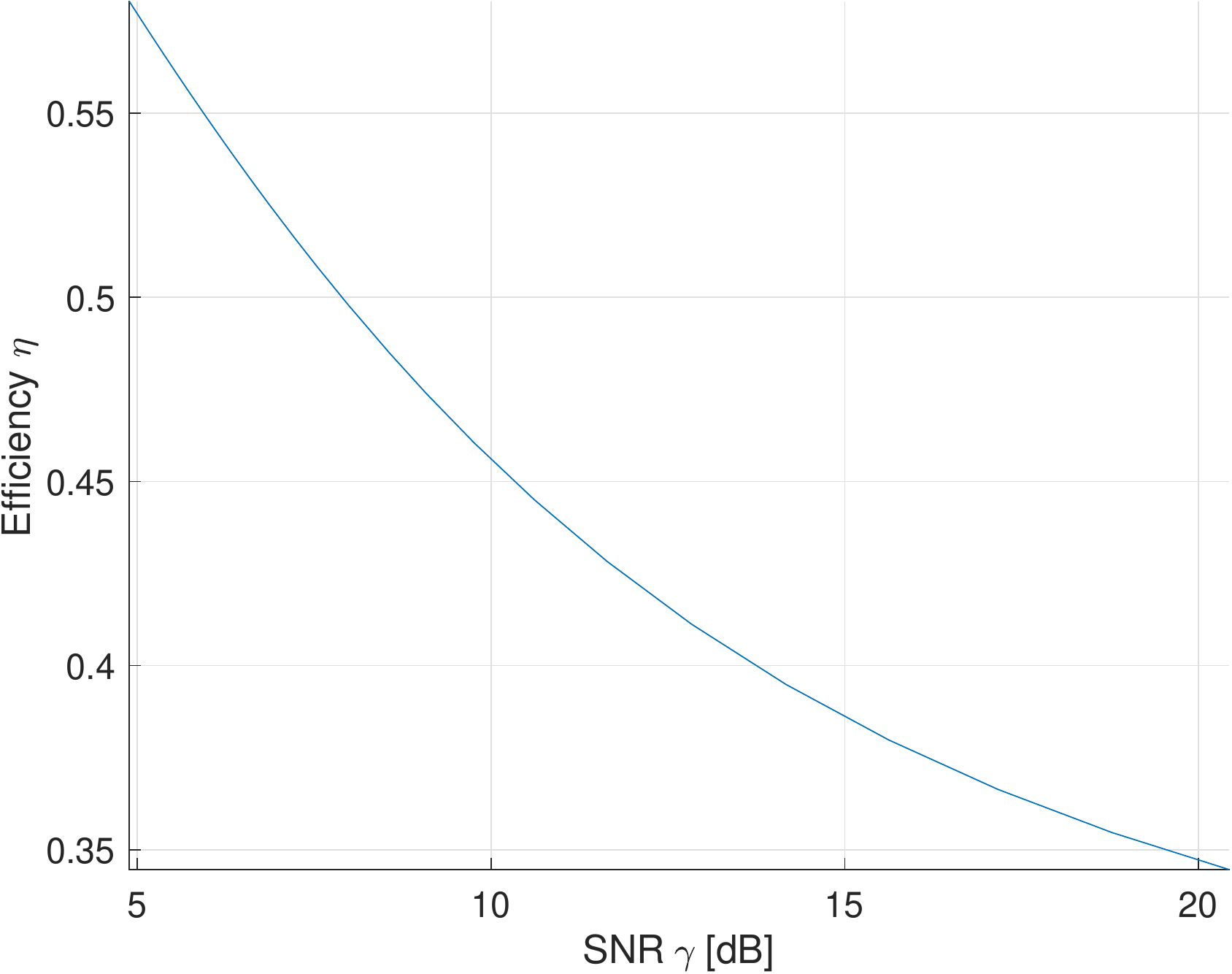}
\caption{The efficiency $\eta$ of the $(4,2)$ stabilizing code of Example 1 as a function of $\gamma$. }
\label{fig:efficiency}
\end{center}
\end{figure}

\subsection{Stabilizing codes based on multiple descriptions}
It is possible to introduce correlation between the quantization noises $q_i, i=1,\dotsc, k,$ of the encodings in Definition~\ref{def:encodings}, which makes it possible to exploit the benefits of multiple descriptions. Of course, zero correlation is a special case of multiple descriptions, which is usually referred to as the no excess marginal rate case~\cite{Zhang1995MultipleDS}. 
%By introducing correlation one can improve upon the performance when receiving a certain amount of descriptions by trading-off the performance, when receiving less descriptions, and vice versa. 
When introducing correlation, the sum-rate is no longer simply just given by the sum of the optimal marginal (description) rates. The sum-rate also becomes a function of the amount of correlation introduced; the greater (negative) correlation, the greater  sum-rate \cite{ozarow:1980}. 

\begin{lemma}
$\mathbf{(k,k')}$ \textbf{Stabilizing Code Based on Multiple Descriptions.}
\emph{Consider the system $(F,L,P,\gamma)$, which is illustrated in Fig.~\ref{fig:linear_system}. Let $v$ be Gaussian and let $w_i = v + q_i, i = 1, \dotsc, k$, 
where $q_i, i=1,\dotsc, k,$ are zero-mean Gaussian distributed with variance $\sigma^2$, and pairwise correlated with correlations coefficient $-\frac{1}{k-1}<\rho \leq 0$. 
If for some $k'$ and $\rho$, the common variance $\sigma^2$ satisfies
\begin{equation}\label{eq:code2}
\sigma^2\leq  \frac{\gamma k'\|L_yP_{21}S\|^{2}}{ \|S-1\|^{2} (\gamma - \|S-1\|^2)(1+(k'-1)\rho)},
\end{equation}
where $S$ is given in \eqref{eq:S}, then $w_i, i=1,\dotsc, k$, form a $(k,k')$ stabilizing code for the system $(F,L,P,\gamma)$. }
\end{lemma}

%\vspace{2mm}

\begin{proof}
We need to ensure that $\mathrm{SNR}>\|S-1\|^2$, when receiving at least $k'$ descriptions. 
The  noise variance when combining any $k'$ descriptions is given by:
\begin{align}\label{eq:var}
\mathrm{var}\bigg( \frac{1}{k'} \sum_{i=1}^{k'} q_i \bigg) = \frac{\sigma^2}{k'}(1+(k'-1)\rho).
\end{align}
Using \eqref{eq:var},  the $\mathrm{SNR}$ is given by:
\begin{align}
\frac{\sigma_v^2}{\frac{\sigma^2}{k'}(1+(k'-1)\rho)} \geq \|S-1\|^2.
\end{align}
Isolating $\sigma^2$ and inserting \eqref{eq:sv} leads to:
\begin{align}
\sigma^2 &\leq k'\|S-1\|^{-2}(1+(k'-1)\rho)^{-1}\sigma_v^{-2} \\ \notag
&= k'\|S-1\|^{-2}(1+(k'-1)\rho)^{-1}
\gamma (\gamma - \|S-1\|^2)^{-1} \\
&\quad\times \|L_yP_{21}S\|^{2}.
\end{align}

\end{proof}

Let $\rho \in (\frac{-1}{k-1},0]$ be the common correlation coefficient between all noise pairs $q_i,q_j, \forall i\neq j$, and let $\sigma^2$ be their common variance. If we are only interested in the performance when receiving $k'$ descriptions or all $k$ descriptions, then the sum-rate $R_S$ can be explicitly expressed \cite{pradhan:2004}:
\begin{align}\notag
R_S  &=\frac{1}{2k'}\log_2 \bigg( \frac{k' + \sigma^2(1+(k'-1)\rho)}{\sigma^2(1-\rho)} \bigg) \\
&\quad + \frac{1}{2k}\log_2 \bigg(\frac{1-\rho}{1+(k-1)\rho} \bigg)\bigg),
\end{align}
where it is assumed the source is standard normal.

\section{Simulation Study}
We consider the same system as that of Example 1, and we will assume i.i.d.\ packet losses.  The encoder is informed about the packet loss probability but does not know when an erasure occur. Knowledge of the packet loss probability makes it possible to design an efficient entropy coder (lossless coder). 

We will be using a subtractively dithered scalar quantizer, which is a stochastic quantizer that provides different outputs, when encoding the same source multiple times \cite{zamir:2014}. We will use this to form the $k$ \emph{independent} encodings.
\begin{figure}[t]
\begin{center}
\includegraphics[width=7cm]{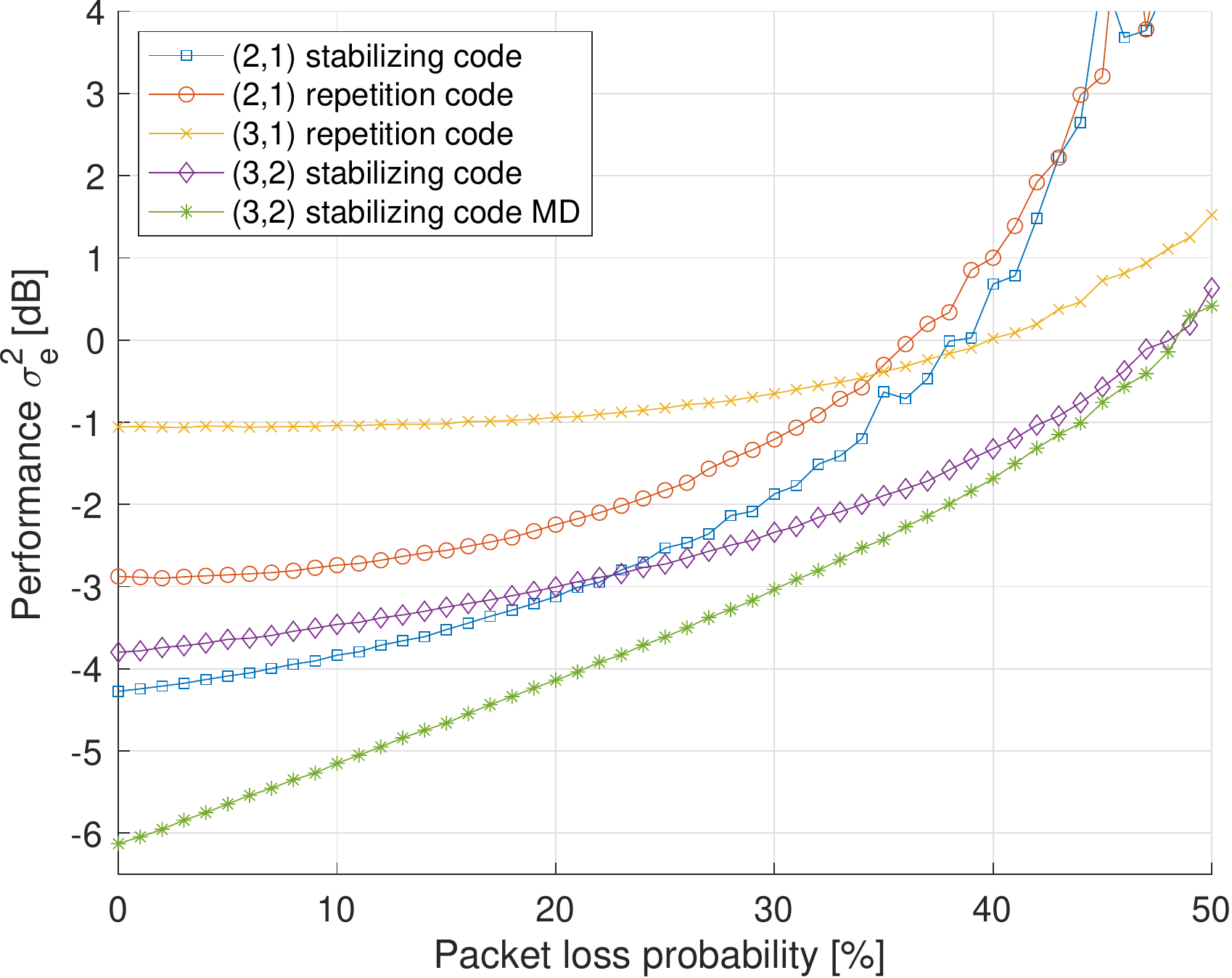}
\caption{ The performance of stabilizing and repetition codes as a function of packet-loss probability.}
\label{fig:performance_all}
\end{center}
\end{figure}

% and use the measured output entropy as an indication of the operational bitrate. S

We encode the output $v$ of Fig.~\ref{fig:linear_system} using a  subtractively dithered scalar quantizer $\mathcal{Q}_\Delta$ with step-size $\Delta$ to obtain: 
\begin{equation}
w_i = \mathcal{Q}_\Delta(v + \xi_i)-\xi_i, i = 1,\dotsc, k. 
\end{equation}
where $\xi$ denotes the dither signal. We choose the step-size $\Delta$ so that the resulting $\mathrm{SNR}$ when using only a single description $w_i$ is $1.76$, which is below that required for stability. Combining any two descriptions yields $\mathrm{SNR} = 3.57$ and combining all three yields $\mathrm{SNR} = 5.29$. Thus, using at least two descriptions is sufficient to stabilize the system. Based on this we design a $(3,2)$ stabilizing code, which for each input sample produces three outputs using the quantizer three times. The theoretical efficiency of this scheme is $\log2(1+ 5.29)/(3\log_2(1+1.76)) = 0.6$. In practice, we suffer from a rate loss due to using a scalar quantizer. The theoretical rate is $1/2\log_2(1 + 1.76) = 0.73$. However, transmitting less than one bit per sample is only possible when encoding vectors. The measured entropy of the quantized output is $1.57$ bits per description.

 We have plotted the performance of the $(3,2)$ stabilizing code in Fig.~\ref{fig:performance_all} as a function of the packet-loss probability. We assume i.i.d.\ packet losses, and simply average the received descriptions to form the reconstruction. Also shown is the performance when transmitting one of the descriptions three times. This corresponds to a $(3,1)$ repetition code having similar sum-rate as the $(3,2)$ stabilizing code. For each packet-loss probability, the performance and rates are averages over a realization having $10^6$ samples. It can be observed that using a stabilizing code is up till $3$ dB better than a repetition code at low packet-loss probabilities.

We also show in Fig.~\ref{fig:performance_all} the performance of a $(2,1)$ stabilizing code, which is compared to a $(2,1)$ repetition code. The $\mathrm{SNR}$ is $3.45$ and $6.91$, when using 1 or 2 descriptions, respectively, of the $(2,1)$ stabilizing code. The measured output entropy after scalar quantization is in this case $1.76$ bits per description, and the sum-rate is $3.52$ bits. 

Finally, we design a $(3,2)$ stabilizing code based on multiple descriptions. It is not straight-forward to obtain correlated noises between the descriptions, and we use here the approach described in~\cite{ostergaard:2006}, which is based on nested lattices and index assignments. The source is first quantized using a fine-grained quantizer referred to as the central quantizer. Then, a one-to-many map is applied, which maps the quantized value to $k$ points in a nested (coarser) lattice. If all $k$ coarser points are received, the map is invertible and the point of the central quantizer is used for reconstruction. If less than $k$ descriptions are received, the reconstruction is given by the average of the received points in the coarser lattice~\cite{ostergaard:2006}. We are using a nesting factor of 5, and the resulting pairwise correlation coefficient between the descriptions is $\rho = -0.41$. The $\mathrm{SNR}$ of a single description is $1.68$ and that of two descriptions is $5.6$, which is above the critical value for stability. The $\mathrm{SNR}$ when all descriptions are used is $12.0$. 
The step-size of the fine lattice is chosen such that the resulting bitrate is similar to that of the $(3,2)$ stabilizing code based on independent encodings. 

It can be seen in Fig.~\ref{fig:performance_all} that stabilizing codes outperforms repetition coding. Moreover, using MD coding when constructing the stabilizing codes is better than using independent encodings, except at very low bitrates or very high packet-loss rates.

%%%%%%%%%%%%%%%%%%%%%%%%%%%%%%%%%%%%%%%%%%%%%%%%%%%%%%%%%%%%%%%%%%%%%%%%%%%%%%%%
\section{Conclusions}
A new construction of error correction codes were proposed, which takes the stability of the plant into account. For linear systems with scalar input and output, explicit designs were provided, and it was shown that there is a significant gain over using traditional repetition codes. Similar to repetition coding, the proposed codes do not add additional delays but operate on each sample at a time.

%Future works include extending the proposed constructions to linear dynamical vector systems and non-linear Markov jump linear systems.

%\addtolength{\textheight}{-6cm}

%%%%%%%%%%%%%%%%%%%%%%%%%%%%%%%%%%%%%%%%%%%%%%%%%%%%%%%%%%%%%%%%%%%%%%%%%%%%%%%%
\section{Acknowledgments}

The author would like to thank Mohsen Barforooshan for discussions related to simulating the control system.

\bibliography{literature} 

\begin{thebibliography}{10}
\providecommand{\url}[1]{#1}
\csname url@samestyle\endcsname
\providecommand{\newblock}{\relax}
\providecommand{\bibinfo}[2]{#2}
\providecommand{\BIBentrySTDinterwordspacing}{\spaceskip=0pt\relax}
\providecommand{\BIBentryALTinterwordstretchfactor}{4}
\providecommand{\BIBentryALTinterwordspacing}{\spaceskip=\fontdimen2\font plus
\BIBentryALTinterwordstretchfactor\fontdimen3\font minus
  \fontdimen4\font\relax}
\providecommand{\BIBforeignlanguage}[2]{{%
\expandafter\ifx\csname l@#1\endcsname\relax
\typeout{** WARNING: IEEEtran.bst: No hyphenation pattern has been}%
\typeout{** loaded for the language `#1'. Using the pattern for}%
\typeout{** the default language instead.}%
\else
\language=\csname l@#1\endcsname
\fi
#2}}
\providecommand{\BIBdecl}{\relax}
\BIBdecl

\bibitem{tatikonda:}
S.~Tatikonda and S.~Mitter, ``Control over noisy channels,'' \emph{IEEE Trans.
  Automatic Control}, 2004.

\bibitem{sinopoli:2004}
B.~Sinopoli, L.~Schenato, M.~Franceschetti, K.~Poolla, M.~Jordan, and
  S.~Sastry, ``Kalman filtering with intermittent observations,'' \emph{IEEE
  Trans. Autom. Control}, 2004.

\bibitem{jin:2006}
R.~M. Z.~Jin, V.~Gupta, ``State estimation over packet dropping networks using
  multiple description coding,'' \emph{IEEE Trans. Autom. Control}, vol.~42,
  no.~9, pp. 1441 -- 1452, 2006.

\bibitem{imer:2006}
O.~Imer, S.~Yüksel, and T.~Başar, ``Optimal control of {LTI} systems over
  unreliable communication links,'' \emph{Automatica}, vol.~42, pp. 1429 --
  1439, 2006.

\bibitem{sahai:2006}
A.~Sahai, ``The necessity and sufficiency of anytime capacity for stabilization
  of a linear system over a noisy communication link - part i: Scalar
  systems,'' \emph{IEEE Transactions on Information Theory}, vol.~52, pp. 3369
  -- 3395, August 2006.

\bibitem{liu:2007}
G.~Liu, Y.~Xia, J.~Chen, D.~Rees, and W.~Hu, ``Networked predictive control of
  systems with random network delays in both forward and feedback channels,''
  \emph{IEEE Trans. Ind. Electron.}, vol.~54, no.~3, pp. 1282 -- 1297, 2007.

\bibitem{quevedo:2007}
D.~E. {Quevedo}, E.~I. {Silva}, and G.~C. {Goodwin}, ``Packetized predictive
  control over erasure channels,'' in \emph{2007 American Control Conference},
  2007, pp. 1003 -- 1008.

\bibitem{ostrovsky:2009}
R.~{Ostrovsky}, Y.~{Rabani}, and L.~J. {Schulman}, ``Error-correcting codes for
  automatic control,'' \emph{IEEE Transactions on Information Theory}, vol.~55,
  no.~7, pp. 2931 -- 2941, 2009.

\bibitem{gupta:2009}
V.~{Gupta}, A.~F. {Dana}, J.~P. {Hespanha}, R.~M. {Murray}, and B.~{Hassibi},
  ``Data transmission over networks for estimation and control,'' \emph{IEEE
  Transactions on Automatic Control}, vol.~54, no.~8, pp. 1807 -- 1819, 2009.

\bibitem{trivellato:2010}
M.~Trivellato and N.~Benvenuto, ``State control in networked control systems
  under packet drops and limited transmission bandwidth,'' \emph{IEEE Trans.
  Communications}, vol.~58, no.~2, pp. 611 -- 622, 2010.

\bibitem{elia:2011}
N.~Elia and J.~Eisenbeis, ``Limitations of linear control over packet drop
  networks,'' \emph{IEEE Trans. Automat. Contr.}, vol.~56, no.~4, pp. 826 --
  841, 2011.

\bibitem{quevedo:2011}
D.~Quevedo, J.~{\O}stergaard, and D.~Nesic, ``Packetized predictive control of
  stochastic systems over bit-rate limited channels with packet loss,''
  \emph{IEEE Transactions on Automatic Control}, vol.~56, no.~12, pp. 2854 --
  2868, 2011.

\bibitem{garone:2012}
E.~Garone, B.~Sinopoli, A.~Goldsmith, and A.~Casavola, ``Lqg control for mimo
  systems over multiple erasure channels with perfect acknowledgment,''
  \emph{IEEE Transactions on Automatic Control}, vol.~57, no.~2, pp. 450 --
  456, 2012.

\bibitem{yuksel:2013}
S.~Yüksel and S.~P. Meyn, ``Random-time, state-dependent stochastic drift for
  markov chains and application to stochastic stabilization over erasure
  channels,'' \emph{IEEE Trans. Autom. Control}, 2013.

\bibitem{quevedo:2014}
D.~Quevedo, J.~{\O}stergaard, and A.~Ahlen, ``Power control and coding
  formulation for state estimation with wireless sensors,'' \emph{IEEE
  Transactions on Control Systems Technology}, vol.~22, pp. 413 -- 427, 2014.

\bibitem{nagahara:2014}
M.~Nagahara, D.~Quevedo, and J.~{\O}stergaard, ``Sparse packetized predictive
  control for networked control over erasure channels,'' \emph{IEEE
  Transactions on Automatic Control}, vol.~59, no.~7, pp. 1899 -- 1905, 2014.

\bibitem{2015IJC}
A.~{Farhadi}, ``{Stability of linear dynamic systems over the packet erasure
  channel: a co-design approach},'' \emph{International Journal of Control},
  vol.~88, no.~12, pp. 2488 -- 2498, Dec. 2015.

\bibitem{ostergaard:2016}
J.~{\O}stergaard and D.~Quevedo, ``Multiple descriptions for packetized
  predictive control,'' \emph{EURASIP J. Adv. Signal Proc.}, vol. 2016, no.~45,
  April 2016.

\bibitem{peters:2016}
E.~Peters, D.~Quevedo, and J.~{\O}stergaard, ``Shaped {Gaussian} dictionaries
  for quantized networked control systems with correlated dropouts,''
  \emph{IEEE Transactions on Signal Processing}, vol.~64, no.~1, pp. 203 --
  213, 2016.

\bibitem{maas:2016}
A.~Maass, F.~Vargas, and E.~Silva, ``Optimal control over multiple erasure
  channels using a data dropout compensation scheme,'' \emph{Automatica},
  vol.~68, pp. 155 -- 161, 2016.

\bibitem{khina:2019}
A.~Khina, V.~Kostina, and A.~K.~B. Hassibi, ``Tracking and control of
  {Gauss-Markov} processes over packet-drop channels with acknowledgments,''
  \emph{IEEE Trans. Control of Network Systems}, vol.~6, no.~2, June 2019.

\bibitem{barforooshan:2020}
M.~Barforooshan, M.~Nagahara, and J.~{\O}stergaard, ``Sparse packetized
  predictive control over communication networks with packet dropouts and time
  delays,'' in \emph{IEEE 58th Conference on Decision and Control (CDC)}, 2020,
  pp. 8272 -- 8277.

\bibitem{singleton:1964}
R.~Singleton, ``Maximum distance q-nary codes,'' \emph{IEEE Trans. Inf.
  Theory}, vol.~10, no.~2, pp. 116 -- 118, 1964.

\bibitem{silva:2016}
E.~Silva, M.~Derpich, J.~{\O}stergaard, and M.~Encina, ``A characterization of
  the minimal average data rate that guarantees a given closed-loop performance
  level,'' \emph{IEEE Transactions on Automatic Control}, vol.~61, no.~8, pp.
  2171 -- 2186, 2016.

\bibitem{gamal:1982}
A.~Gamal and T.~Cover, ``Achievable rates for multiple descriptions,''
  \emph{IEEE Trans. Inf. Theory}, vol. IT-28, no.~6, pp. 851 -- 857, 1982.

\bibitem{yeung:1996}
R.~Yeung and R.~Zamir, ``Multilevel diversity coding via successive
  refinement,'' in \emph{Proceedings of IEEE International Symposium on
  Information Theory}, 1996.

\bibitem{puri:1999}
R.~Puri and K.~Ramchandran, ``Multiple description source coding using forward
  error correction codes,'' in \emph{Asilomar Conf Signals Syst. Comput.},
  1999.

\bibitem{ostergaard:2021}
J.~{\O}stergaard, U.~Erez, and R.~Zamir, ``Incremental refinements and multiple
  descriptions with feedback,'' \emph{Submitted to IEEE Transactions on
  Information Theory}, 2020, electronically available on arxiv.org:
  https://arxiv.org/abs/2011.02747.

\bibitem{zamir:2014}
R.~Zamir, \emph{Lattice Coding for Signals and Networks A Structured Coding
  Approach to Quantization, Modulation and Multiuser Information Theory}.\hskip
  1em plus 0.5em minus 0.4em\relax Cambridge University Press, 2014.

\bibitem{ostergaard:2021b}
J.~{\O}stergaard, ``Stabilizing error correction codes for control over erasure
  channels,'' 2021, submitted to IEEE Transactions on Control of Network
  Systems. Draft available at: https://arxiv.org/abs/2112.11717.

\bibitem{slepian:1973}
D.~Slepian and J.~Wolf, ``Noiseless coding of correlated information sources,''
  \emph{IEEE Transactions on Information Theory}, vol.~19, no.~4, pp. 471 --
  480, 1973.

\bibitem{Zhang1995MultipleDS}
Z.~Zhang and T.~Berger, ``Multiple description source coding with no excess
  marginal rate,'' \emph{IEEE Trans. Inf. Theory}, vol.~41, pp. 349--357, 1995.

\bibitem{ozarow:1980}
L.~Ozarow, ``On a source-coding problem with two channels and three
  receivers,'' \emph{Bell Syst. Tech. J.}, vol.~59, no.~10, pp. 1909 -- 1921,
  1980.

\bibitem{pradhan:2004}
S.~Pradhan, R.~Puri, and K.~Ramchandran, ``$n$-channel symmetric multiple
  descriptions - part i: $(n,k)$ source-channel erasure codes,'' \emph{IEEE
  Trans. Inf. Theory}, vol.~50, no.~1, pp. 47 -- 61, January 2004.

\bibitem{ostergaard:2006}
J.~{\O}stergaard, J.~Jensen, and R.~Heusdens, ``$n$-channel entropy-constrained
  multiple-description lattice vector quantization,'' \emph{IEEE Transactions
  on Information Theory}, no.~5, pp. 1956 -- 1973, 2006.

\end{thebibliography}
\bibliographystyle{IEEEtran}

\end{document}